\documentclass[submission,copyright,creativecommons]{eptcs}
\providecommand{\event}{ICLP 2023} 

\usepackage{iftex}

\usepackage[frozencache=true, cachedir=minted-cache]{minted}
\usepackage{amsmath}
\usepackage{amsthm}
\usepackage{graphicx}
\usepackage{multirow}
\usepackage{amssymb}
\usepackage{adjustbox}
\usepackage{tikz-cd}
\usepackage{fixltx2e}
\usepackage{algorithm}
\usepackage[noend]{algpseudocode}

\ifpdf
  \usepackage{underscore}         
  \usepackage[T1]{fontenc}        
\else
  \usepackage{breakurl}           
\fi

\title{On the Independencies Hidden in the Structure of a Probabilistic Logic Program}
\author{Kilian Rückschloß
\institute{ Ludwig-Maximilians-Universität München \\
			Oettingenstra\ss{e} 67, 80538 München, Germany\\
         \email{kilian.rueckschloss@lmu.de}}
\and
Felix Weitkämper 
\institute{ Ludwig-Maximilians-Universität München \\
			Oettingenstra\ss{e} 67, 80538 München, Germany\\
\email{felix.weitkaemper@lmu.de}
}
}
\newcommand{\titlerunning}{On the Independencies Hidden in the Structure of a PLP}
\newcommand{\authorrunning}{K.~Rückschloß, F.~Weitkämper}

\hypersetup{
  bookmarksnumbered,
  pdftitle    = {\titlerunning},
  pdfauthor   = {\authorrunning},
  pdfsubject  = {EPTCS},               
}

\newtheorem{definition}{Definition}

\newtheorem{lemma}{Lemma}

\newtheorem{theorem}{Theorem}
\newtheorem{example}{Example}
\newtheorem{proposition}{Proposition}
\newtheorem{program}{Program}

\DeclareMathOperator{\pa}{pa}
\DeclareMathOperator{\var}{var}

\DeclareMathOperator{\effect}{effect}
\DeclareMathOperator{\causes}{causes}
\DeclareMathOperator{\condition}{cond}

\DeclareMathOperator{\Graph}{Graph}

\DeclareMathOperator{\Pa}{Pa}

\begin{document}
\maketitle

\begin{abstract}
Pearl and Verma developed d-separation as a widely used graphical criterion to reason about the conditional independencies that are implied by the causal structure of a Bayesian network. As acyclic ground probabilistic logic programs correspond to Bayesian networks on their dependency graph, we can compute conditional independencies from d-separation in the latter.

In the present paper, we generalize the reasoning above to the non-ground case. First, we abstract the notion of a probabilistic logic program away from external databases and probabilities to obtain so-called program structures. We then present a correct meta-interpreter that decides whether a certain conditional independence statement is implied by a program structure on a given external database. Finally, we give a fragment of program structures for which we obtain a completeness statement of our conditional independence oracle. We close with an experimental evaluation of our approach revealing that our meta-interpreter performs significantly faster than checking the definition of independence using exact inference in ProbLog 2. 
\end{abstract}

\section{Introduction}

A probabilistic logic program is a logic program, in which each clause holds with a specified probability. The most common semantics for these programs is the distribution semantics \cite{DistributionSemantics}, which assigns to each ground program a joint probability distribution over the atoms occurring in it. It is the basis for many programming languages such as the  Independent Choice Logic \cite{IndependentChoiceLogic}, PRISM \mbox{\cite{PRISM}}, Logic Programs with Annotated Disjunctions \cite{LPAD} and ProbLog \cite{Problog}.

Since conditional independence is a rather fundamental notion in probability theory, it is natural to ask for its counterpart in probabilistic logic programming. Moreover, considering the work of Holtzen et al.~\cite{dice}, such an analysis may contribute to speed up lifted inference. In this paper, we extend the effort of Rückschloß and Weitkämper~\cite{ilp2021} to establish a calculus deriving the conditional independencies that are determined by the clause structure of a ProbLog program. Let us illustrate this problem in the following example:

\begin{example}
We consider storages, which consist of rooms, tanks, employees and liquids that are given by the unary predicates $room/1$, $tank/1$, $employee/1$, $liquid/1$, respectively. For each storage a database table $passage (R,R')$ tells us, between which rooms $R$ and $R'$ we find a passage. Moreover, in each room~$R$ we find tanks $T$, denoted by~$in (T,R)$ and a tank $T$ may store a liquid $L$, denoted by $stores (T,L)$. Finally, we assume a database table~$flammable(L)$, indicating the flammable liquids. In this context, it makes sense to assume that each tank stores at most one liquid which results in the following integrity constraint:
\begin{align}
\tag{Const} \label{example - integrity constaint}
\perp \leftarrow tank(T), liquid(L_1), liquid(L_2), L_1 \not\doteq L_2, stores(T,L_1), stores(T,L_2).
\end{align}

With a certain probability $\_$ we find an employee $E$ opening a tank $T$, denoted by~$opens (E,T)$. This leads to the following ProbLog clause:
\begin{align}\tag{RC1}
\_ :: opens(E,T) \leftarrow employee(E), tank(T).
\label{example - opens}
\end{align} 
Further, if employee $E$ opens tank $T$, it may be that $E$ does not close the tank $T$ properly, which then causes the tank $T$ to leak, denoted by $leaks (T)$. Again we assume that we don't know the exact probabilities and capture the mechanism in the following clause:
\begin{align}\tag{RC2}
\_ :: leaks(T) \leftarrow employee(E), tank(T), opens(E,T).
\label{example - leaks}
\end{align}   
   
Moreover, an employee $E$ may smoke in a room $R$, denoted by $smokes (E,R)$:
\begin{align}\tag{RC3}
\_ :: smokes(E,R) \leftarrow employee(E), room(R).
\label{example - smokes}
\end{align}
If employee $E$ smokes in a room $R$, which contains a leaking tank storing an flammable liquid, this may cause a fire in room $R$. Further, the smoke immediately spreads to the rooms $R_1$, which are connected to the room $R$ by passages, where it may trigger a sensor. In this case, we observe the event $fire (R_1)$. This mechanism is captured in the clauses  below:
\begin{align}
& connected(R,R) \leftarrow room(R). \nonumber \\
& connected(R,R_1) \leftarrow room(R), room(R_1), room(R_2), passage(R_2,R_1), \nonumber \\ 
& \text{ }~~~~connected(R, R_2). \tag{Int}\label{example - connected} \\
& \_ :: fire(R_1) \leftarrow room(R_1), room(R), employee(E), tank(T), liquid(L), in(T,R), \nonumber \\ 		
& \text{	}~~~~(connected(R,R_1); connected(R_1,R)), stores(T,L), flammable(L),  \nonumber \\
& \text{	}~~~~smokes(E,R), leaks(T). \tag{RC4}\label{example - fire}
\end{align}

Given (\ref{example - opens}), (\ref{example - leaks}), (\ref{example - smokes}), (\ref{example - fire}), (\ref{example - integrity constaint}) and~(\ref{example - connected}) together with a concrete storage, we aim to answer queries about conditional independencies. For instance, we want to determine whether the event $smokes(e1,r)$ of employee $e1$ smoking in room $r$ is independent of the event~$opens(e2,t)$ of employee~$e2$ opening the tank $t$ without knowing the probabilities of the clauses~\mbox{(\ref{example - opens}), (\ref{example - leaks}), (\ref{example - smokes})} and~(\ref{example - fire}). Further, we would also like to predict how things change if we observe $fire(r_1)$ a fire in room~$r_1$.  
\label{example - formulation}
\end{example}

We accomplish our goal by grounding ProbLog programs to Bayesian networks, where we can apply the theory of d-separation \cite{d-separation} to derive the desired independence statements. In this way we can efficiently reason about conditional independencies. Further, we highlight that the theory of d-separation lies at the basis of constraint based causal structure discovery, i.e.~it enables us to reason about causal relationships on the basis of observational data. Hence, we suppose that the present work also serves as a starting point in the development of causal structure discovery techniques for probabilistic logic programming.

\section{\mbox{On the Independencies Hidden in a Propositional Causal Structure}}
\label{section - the BN case}

At the beginning, we discuss how conditional independencies can be inferred from a causal structure in the propositional case. Here, we identify a \textbf{causal structure} on a set of random variables $\textbf{V}$ with a directed acyclic graph $G$, i.e.~a partial order, on $\textbf{V}$. The intuition is that $X$ is a \textbf{cause} of $Y$ if there is a directed path from $X$ to $Y$ in $G$. In this case, we also say that~$Y$ is an \textbf{effect} of $X$. Further, we say that~$X$ is a \textbf{direct cause} of~$Y$ if the edge $X \rightarrow Y$ exists in $G$, i.e. if and only if the node $X \in \Pa(Y)$ lies in the set~$\Pa(Y)$ of \textbf{parents} of~$Y$. 

\begin{example}
Consider a road that passes by a field with a sprinkler in it. The sprinkler is switched on by a weather sensor and the pavement of the road may be wet, denoted by $wet$, because the sprinkler is on, denoted $sprinkler$ or because it rains, denoted by $rain$. Further, we know that the events $rain$ and $sprinkler$ are caused by the season, denoted by $season$, as they both are triggered by the weather. Finally, we note that a wet road is more likely to be slippery, denoted by $slippery$. The situation above gives rise to the following causal structure on the random variables \mbox{$\textbf{V} := \{ season, rain, sprinkler, wet, slippery \}$}:
\begin{equation}
\begin{tikzcd}
                             & rain \arrow[rd]      &               &          \\
season \arrow[ru] \arrow[r] & sprinkler \arrow[r] & wet \arrow[r] & slippery 
\end{tikzcd}
\label{example - sprinkler}
\end{equation}
In particular, we find that $season$ is a cause of slippery but not a direct cause, whereas $wet$ is a direct cause of $slippery$. Further, there is no causal relationship between $sprinkler$ and $rain$.
\label{example - causal structure}
\end{example} 

Next, a given probability distribution $\pi$ on the random variables $\textbf{V}$ is consistent with a causal structure~$G$ if the influence of any cause $X$ on an effect $Y$ is moderated by the direct causes of $Y$. This intuition is formally captured in the following definition:

\begin{definition}[Markov Condition]
We say that the distribution $\pi$ on the set of random variables $\textbf{V}$ satisfies the \textbf{Markov condition} with respect to a causal structure $G$ on $\textbf{V}$, if every random variable $X \in \textbf{V}$ is independent of its causes in $G$, once we observe its direct causes $\Pa (X)$. In this case, we write $\pi \models G$.
\end{definition} 

\begin{example}
In Example \ref{example - causal structure} the Markov condition states for instance that the influence of $season$ on the event $slippery$ is completely moderated by the event $wet$. Once we know that the pavement of the road is wet, we expect it to be slippery regardless of the event that caused the road to be wet. 
\end{example}

If a distribution $\pi \models G$ satisfies the Markov condition with respect to a given causal structure $G$, it is represented by a Bayesian network on $G$ and vice versa~\mbox{\cite[§1.2.3]{Causality}}:

\begin{definition}[Bayesian Network]
A \textbf{Bayesian network} on a set of random variables~$\textbf{V}$ consists of a causal structure~$G$ on $\textbf{V}$ and the probability distributions $\pi (X \vert \Pa (X))$ of the random variables $X \in \textbf{V}$ conditioned on their direct causes in~$G$. A Bayesian network gives rise to a joint probability distribution on~\mbox{$\textbf{V} = \{ X_1, ... ,X_k \}$} by setting
$$
\pi \left( X_1 = x_1,..., X_k = x_k \right) :=
\prod_{i=1}^k \pi \left( X_i = x_i \vert \pa(X_i) \right), 
$$   
where~\mbox{$\pa(X_i) := \{ X_j = x_j \vert X_j \in \Pa (X_i) \}$}. 
\label{definition - Bayesian network}
\end{definition}

The Markov condition equips a causal structure with a semantics that is given by conditional independence statements. Further, Verma and Pearl \cite{d-separation} derive d-separation as a criterion to compute all conditional independencies that follow if we apply the Markov condition to a given causal structure.

\begin{definition}[d-Separation]
Let $G$ be a directed acyclic graph, i.e.~it is a causal structure. An \textbf{undirected path} $P$ between two nodes $A$ and $B$ is an alternating sequence of nodes and edges
$$
P = R_0 \stackrel{E_1}{-} R_1 \stackrel{E_2}{-} R_2 \stackrel{E_3}{-} ... \stackrel{E_{n-1}}{-} R_{n-1} \stackrel{E_n}{-} R_n,
$$ 
where $E_i \in \{ R_{i-1} \rightarrow R_i \text{, } R_{i-1} \leftarrow R_i \}$ for all $1 \leq i \leq n$. We call a node $R_i$ of $P$ a \textbf{collider} if $P$ is of the form~\mbox{$... \rightarrow R_i \leftarrow ...$}, otherwise $R_i$ is said to be a \textbf{non-collider} of $P$.
 
Further, let $\textbf{Z}$ be a set of nodes. We say that a node $N$ is \textbf{blocked} by~$\textbf{Z}$ if it lies in~$\textbf{Z}$, i.e.~if we have that $N \in \textbf{Z}$. Moreover, $N$ is said to be \textbf{activated} by $\textbf{Z}$ if there exists a directed path from $N$ to a node in $\textbf{Z}$. The undirected path $P$ is a \textbf{d-connecting path} with respect to the \textbf{observations} $\textbf{Z}$ if every non-collider $N$ of $P$ is not blocked and if
every collider~$C$ of~$P$ is activated.
  
We say that $\textbf{Z}$ \textbf{d-connects} two nodes $A$ and $B$ if there exists a d-connecting path between~$A$ and~$B$ with respect to $\textbf{Z}$. Otherwise, we say that $\textbf{Z}$ \textbf{d-separates} $A$ and~$B$. Finally, two sets of nodes $\textbf{A}$ and $\textbf{B}$ are said to be \textbf{d-separated} by $\textbf{Z}$ if $\textbf{Z}$ d-separates~$A$ and~$B$ for every $A \in \textbf{A}$ and every $B \in \textbf{B}$. Otherwise, the sets $\textbf{A}$ and $\textbf{B}$ are \textbf{d-connected} by $\textbf{Z}$.
\label{definition - d-connecting path}
\end{definition}

Note that the term ``d-connected'' is a shorthand for ``directionally connected''~\mbox{\cite{IdentifyingIndependenceInBayesianNetworks}}.

\begin{example}
Let us consider the causal structure (\ref{example - sprinkler}) again and take $\textbf{Z} := \{ slippery \}$ for the observations. We find the following \mbox{d-connecting path}~$P := season \rightarrow rain \rightarrow wet \leftarrow sprinkler$. The intuition behind this d-connecting path is as follows:

Assume we observe the event $\textbf{Z}$. We know that this increases the probability for~$wet$, which itself is triggered by $rain$ or $sprinkler$. If we additionally suppose that it is summer, this decreases the probability for $rain$, which increases the probability of $sprinkler$ as we have an increased probability for $wet$. To summarize we expect $season$ and $sprinkler$ to be dependent once we observed $slippery$.

Note that the argument above does not go through anymore, if we observe additionally that it rains or if we observe nothing.  
\label{example - d - separation}
\end{example}

The reasoning of Example \ref{example - d - separation} is now formalized in the following theorem.

\begin{theorem}[Verma and Pearl \cite{d-separation}]
Let $G$ be a causal structure on the set $\textbf{V}$, let ${\textbf{Z} := \{ Z_1,...,Z_n \} \subseteq \textbf{V}}$ be a subset of nodes and let $A,B \in \textbf{V}$ be nodes of $G$.
If $\textbf{Z}$ d-separates $A$ and $B$, we obtain that~$A$ and~$B$ are independent conditioned on $\textbf{Z}$ in every distribution $\pi \models G$, which is Markov to $G$. Here, being conditionally independent means that 
\begin{align}
& \forall_{a \text{ value of } A} \text{ } \forall_{b \text{ value of } B} \text{ }
\forall_{z_1,...,z_n \text{ values of } Z_1,...,Z_n} \text{ : } \nonumber \\ 
& \pi(A = a, B = b \vert \{ Z_i = z_i \}_{i=1}^n) =
\pi(A = a \vert \{ Z_i = z_i \}_{i=1}^n) \cdot
\pi(B = b \vert \{ Z_i = z_i \}_{i=1}^n).~\square  
\label{equation - definition of conditional independence}
\end{align}
\label{theorem - correctness of d-separation}
\end{theorem}

However, McDermott \cite{FailureOfFaithfulness} demonstrates that in general d-separation does not yield a complete independence oracle. This observation motivates the following definition.

\begin{definition}[Faithfulness]
A distribution $\pi$ is \textbf{(causally) faithful} to a causal structure $G$ on $\textbf{V}$ if it is Markov to~$G$ and if every conditional independence of two random variables~\mbox{$A,B \in \textbf{V}$} with respect to a set of observations $\textbf{Z} \subseteq \textbf{V}$ can be derived from d-separation by Theorem~\ref{theorem - correctness of d-separation}.
\label{definition : faithful Bayesian network}
\end{definition}

Fortunately, Meek \cite{FaithfulnessHoldsAlmostAlways} shows that faithfulness holds for almost all Boolean Bayesian networks in the following sense:

\begin{theorem}
Let $G$ be a causal structure and let $\theta \in [0,1]^n$ be the vector, which determines the conditional distributions that turn $G$ into a Boolean Bayesian network representing the distribution~$\pi$. In this case we obtain finitely many non-trivial polynomial equations such that~$\pi$ is faithful to $G$ unless~$\theta$ solves one of these equations. $\square$    
\label{theorem - obstruction to causal faithfulness}
\end{theorem} 

Note that Theorem \ref{theorem - obstruction to causal faithfulness} states that d-separation enables us to derive all conditional independence statements that are implied by a causal structure under the Markov condition. 

Finally, we identify a causal structure $G$ with the database that contains a fact~\mintinline{prolog}{X ---> Y} for every edge $X \rightarrow Y$ in $G$. In this case the predicate \mintinline{prolog}{dseparates/3} in following meta-interpreter decides whether two nodes $X$ and $Y$ are d-separated by a list of observations $\textbf{Z}$.

\begin{program}[Deciding d-Separation]
$\text{ }$
\begin{minted}[mathescape]{prolog}
% Implement hactivates/2 as the transitive closure of (--->)/2 and 
% calculate activated nodes
hactivates(X,X).  hactivates(Z,X):-(Y ---> Z), hactivates(Y,X).
activates([Z|_],X):-hactivates(Z,X).  activates([_|TZ],X):-activates(TZ,X).
% Implement dconnects/3 by case distinction over the last orientation
dconnects(X,Y,Z):-\+member(X,Z),  \+member(Y,Z), dconnects(X,Y,Z,_).
dconnects(X,Y,_,right):-(X--->Y).  dconnects(X,Y,_,left):-(Y--->X).
dconnects(X,Y,Z,right):-(Y1--->Y),\+member(Y1,Z),dconnects(X,Y1,Z,right).
dconnects(X,Y,Z,right):-(Y1--->Y),\+member(Y1,Z),dconnects(X,Y1,Z,left).
dconnects(X,Y,Z,left):-(Y--->Y1),\+member(Y1,Z),dconnects(X,Y1,Z,left).
dconnects(X,Y,Z,left):-(Y--->Y1),activates(Z,Y1),dconnects(X,Y1,Z,right).
% dseparates/3 is the complement of dconnects/3
dseparates(X,Y,Z):-\+dconnects(X,Y,Z).
\end{minted}
\label{program - dconnectes}
\end{program}


\section{A Formalism for Lifted Probabilistic Logic Programming}
\label{Formalism for our Solution}

Recall that events of the trivial probabilities zero and one are independent of every other event. To overcome this obstruction, we introduce a language which separates logical predicates, denoting logical statements with probabilities zero and one, from random predicates, denoting events with a probability lying between zero and one. 

Let us fix a \textbf{query language} i.e.~a language $\mathfrak{Q} \supseteq \mathfrak{L}\supseteq \mathfrak{E}$ in three parts with an \textbf{external vocabulary}~$\mathfrak{E}$ and a \textbf{logical vocabulary}~$\mathfrak{L}$. Here, $\mathfrak{Q}$ is a finite relational vocabulary with equality $\doteq$, i.e.~it consists of a finite set of relation symbols, a finite set of constants as well as a countably infinite set of variables. Further,
$\mathfrak{L}$ is a subvocabulary of  $\mathfrak{Q}$ containing all of the variables and constants of $\mathfrak{Q}$ as well as a (possibly empty) subset of the relation symbols of $\mathfrak{Q}$. Moreover, $\mathfrak{E}$ is a subvocabulary of $\mathfrak{L}$, which satisfies the same properties in $\mathfrak{L}$ as $\mathfrak{L}$ does regarding $\mathfrak{Q}$.

\begin{example}
In Example \ref{example - formulation} the vocabulary $\mathfrak{E}$ consists of the predicates $room/1$, $employee/1$, $tank/1$, $liquid/1$, $passage/2$, $in/2$, $stores/2$ and $flammable/1$, which we assume to be given by a database. Further, $\mathfrak{L}$ extends $\mathfrak{E}$ by the predicate $connected/2$, which is deterministically defined in terms of the predicates of $\mathfrak{E}$. Finally, $\mathfrak{Q}$ extends $\mathfrak{L}$ by the predicates $opens/2$, $leaks/1$, $smokes/2$ and $fire/1$, which we expect to denote non-deterministic random variables.
\label{example - domain vocabulary} 
\end{example}

As usual, an \textbf{atom} is an expression of the form $r(t_1, \dots , t_n)$ or $t_1 \doteq t_2$, where $r$ is a relation symbol and $t_1$ to $t_n$ are constants or variables, and a \textbf{literal} is an expression of the form $A$ or $\neg A$ for an atom $A$. 
It is called an \textbf{external atom} or \textbf{literal} if $r$ is in $\mathfrak{E}$, a \textbf{logical atom} or \textbf{literal} if $r$ is in~$\mathfrak{L}$, an \textbf{internal atom} or \textbf{literal} if $r$ is in~$\mathfrak{L} \setminus \mathfrak{E}$ and a \textbf{random atom} or \textbf{literal} if $r$ is in $\mathfrak{Q} \setminus \mathfrak{L}$. Here, we regard equality $\doteq$ as a relation in $\mathfrak{E}$.
A literal of the form~$A$ is called \textbf{positive} and a literal of the form  $\neg A$ is called \textbf{negative}. A literal $L$ is said to be \textbf{ground} if no variable occurs in it. Finally, we use $\var(E)$ to refer to the variables occurring in a given expression $E$.

\begin{example}
In Example \ref{example - domain vocabulary} $passage (R, R')$ is an external atom, whereas $connected (R,R')$ is an internal atom and $fire(R)$ is a random atom. 
\end{example}

\textbf{Formulas}, as well as \textbf{existential} and \textbf{universal formulas} are defined as usual in \mbox{first-order} logic.
The logical vocabulary will be used to formulate constraints and conditions for our~probabilistic logic program. The purpose of a probabilistic logic program, however, is to define distributions for the random variables determined by the language~$\mathfrak{Q}$. This is done by so-called ProbLog clauses.

\begin{definition}[ProbLog Clause]
A \textbf{generalized ProbLog clause $RC$} is an expression of the form 
$$
\left(\_ :: R \leftarrow R_1,...,R_m, L_1,...,L_n. \right)
=
\left(
\_ :: \effect(RC) \leftarrow \causes(RC) \cup \condition (RC)
\right),
$$
which is given by the following data:
\begin{enumerate}
\item[i)]
a random atom $R := \effect (RC)$, called the \textbf{effect} of $RC$
\item[ii)]
a finite and possibly empty set of random literals $\causes (RC) := \{ R_1,...,R_m \}$, called the \textbf{causes} of~$RC$
\item[iii)]
a finite and possibly empty set of logical literals $\condition (RC) := \{ L_1 , ... , L_n \}$, called the \textbf{condition} of~$RC$   
\end{enumerate}
We call $RC$ \textbf{positive} if the set of causes $\causes(RC)$ contains only positive literals. Further, we obtain a \textbf{ProbLog clause}  
\mbox{$
RC^{\pi(RC)} := \left(\pi (RC) :: R  \leftarrow R_1,...,R_m, L_1,...,L_n.\right)
$}
from the generalized ProbLog clause~$RC$ by choosing a \textbf{probability} $\pi (RC) \in [0,1]$.
\label{definition - random clause}
\end{definition}

In i) and ii) of Definition \ref{definition - random clause} we use the terminology of cause and effect to reflect that under our semantics, a ground program represents a functional causal model~\mbox{\cite[§1.4]{Causality}}. 

\begin{example}
Note that (\ref{example - fire}) of Example~\ref{example - formulation} yields a generalized ProbLog clause. Further, if we choose the probability $\pi(RC4) := 0.6$ we obtain the ProbLog clause $RC4^{0,6}$ below.
\begin{align}
&0.6 :: fire(R1) \leftarrow room(R1), room(R), employee(E), tank(T), liquid(L), in(T,R), \nonumber \\ 		
& \text{	}~~~~(connected(R,R1); connected(R1,R)), stores(T,L), flammable(L),  \nonumber \\
& \text{	}~~~~smokes(E,R), leaks(T). 
\label{example - generalized random clause}
\end{align}
\end{example}

After having established the necessary syntax we proceed to the semantics. Let us begin with the logical expressions. The semantics of logical expressions is given in a straightforward way.

We highlight the unique names assumption in our definition of a structure:

An \textbf{$\mathfrak{L}$-structure} $\Lambda$ consists of a domain $\Delta$, an element of $\Delta$ for every constant in $\mathfrak{L}$, \emph{such that two different constants are mapped to different elements}, and an $n$-ary relation on $\Delta$ for every relation symbol of arity $n$ in~$\mathfrak{L}$.

Whether a logical formula is \textbf{satisfied} by a given $\mathfrak{L}$-structure (under a given \textbf{interpretation} of its free variables) is determined by the usual rules of first-order logic. Finally, note that the semantics of external expressions is defined analogously.

For the semantics of clauses and programs we choose the FCM-semantics \cite{fcm-semantics} since it directly relates an acyclic ground program to a Bayesian network. We start with the definition of a lifted program:

\begin{definition}[Lifted Program and Program Structure]
A \textbf{program structure} $\textbf{P} := (\textbf{R}, \textbf{I}, \textbf{C})$ is a triple, which consists of the following data
\begin{enumerate}
\item[i)] 
a finite set of integrity constraints $\textbf{C}(\textbf{P}) := \textbf{C}$ of the form 
$(\perp \leftarrow L_1, ... , L_k.)$ for logical literals~$L_i$ with $1 \leq i \leq k$, which we call the \textbf{constraints} of $\textbf{P}$. 
\item[ii)]
a finite set of normal clauses $\textbf{I} (\textbf{P}) := \textbf{I}$ of the form 
$(H \leftarrow B_1, ... , B_m.)$ with logical literals~$B_1, ...,B_m$ and an internal atom $H$, which we call the \textbf{internal part}~of~$\textbf{P}$.
\item[iii)]
a finite set of generalized ProbLog clauses $\textbf{R} (\textbf{P}) := \textbf{R}$, which we call the \textbf{random part}~of~$\textbf{P}$.
\end{enumerate}
We say that $\textbf{P}$ is \textbf{stratified} if its internal part $\textbf{I}$ is a stratified set of normal clauses and we say that $\textbf{P}$ is \textbf{positive} if every generalized ProbLog clause of its random part $\textbf{R}(\textbf{P})$ is positive.

A \textbf{choice of parameters} for the program structure $\textbf{P}$ is a function~\mbox{$\pi : \textbf{R} (\textbf{P}) \rightarrow~[0,1]$}. A program structure $\textbf{P}$ and a choice of parameters $\pi$ yield a \textbf{(lifted) program}~$\textbf{P}^{\pi}~:=~(\textbf{R} (\textbf{P})^{\pi}, \textbf{I} (\textbf{P}), \textbf{C} (\textbf{P}))$, where~$\textbf{R} (\textbf{P})^{\pi}$ is the set of ProbLog clauses~\mbox{$
\pi (RC) :: \effect (RC) \leftarrow \causes (RC) \cup \condition (RC) \text{ for } RC \in \textbf{R} (\textbf{P}).
$} 
In this case, $\textbf{C}(\textbf{P})$ are the \textbf{constraints}, $\textbf{I} (\textbf{P})$ is the \textbf{internal part} and $\textbf{R} (\textbf{P})^{\pi}$ is the \textbf{random part} of the program~$\textbf{P}^{\pi}$. Further,~$\textbf{P}$ is called the \textbf{structure} of the program~$\textbf{P}^{\pi}$. Finally, the program~$\textbf{P}^{\pi}$ is \textbf{stratified} or \textbf{positive} if~$\textbf{P}$ is.      
\label{definition - program structure}
\end{definition}

\begin{example}
In Example \ref{example - formulation} we obtain a stratified program structure $\textbf{P}$ with random part (\ref{example - opens}), (\ref{example - leaks}), (\ref{example - smokes}), (\ref{example - fire}), internal part (\ref{example - connected}) and constraints (\ref{example - integrity constaint}).
Further, we obtain a choice of parameters $\pi$ by assigning $\pi (RC1) := 0.8$, $\pi (RC2) := 0.1$, $\pi (RC3) := 0.5$ and $\pi (RC4) := 0.05$. Finally, the choice of parameters $\pi$ gives rise to the following (lifted) program~$\textbf{P}^{\pi}$:
\begin{minted}{prolog}
%Random part
0.8 :: opens(E,T) :- employee (E), tank(T).
0.1 :: leaks(T) :- employee(E), tank(T), opens(E,T).
0.5 :: smokes(E,R) :- employee(E), room(R).
0.05 :: fire(R1) :- employee(E), room(R), room(R1), tank(T), liquid(L), 
	flammable(L), in(T,R), stores(T,L), (connected(R,R1); connected(R1,R)), 
	smokes(E,R),leaks(T).
%Internal part
connected(R,R):-room(R). 
connected(R,R1) :- room(R), room(R1), room(R2), passage(R2,R1), 
	connected(R, R2).
%Constraints
:- tank(T), liquid(L1), liquid(L2), L1 \= L2, stores(T,L1), stores(T,L2).
\end{minted} 
\label{example - program structure}
\end{example}

\begin{definition}[Ground Variable and External Database]
Let $\textbf{P}$ be a stratified program structure and let $\mathcal{E}$ be an $\mathfrak{E}$-structure. In our setting, we may assume without loss of generality that $\mathcal{E}$ is a Herbrand model of a language~$\mathfrak{E}^{\ast}$, which extends the external language $\mathfrak{E}$ by constants. Further, denote by $\mathfrak{L}^*$ and~$\mathfrak{Q}^*$ respectively the extension of the languages $\mathfrak{L}$ and $\mathfrak{Q}$ by the new constants in $\mathfrak{E}^*$. 
We write~\mbox{$
\mathcal{E}^{\textbf{I}} := \mathcal{E}^{\textbf{I}(\textbf{P})} := \{ L \text{ ground atom of } \mathfrak{L}^{\ast} \text{ : } \textbf{I}  \cup \mathcal{E} \models L \}
$} 
for the minimal Herbrand model of $\textbf{I} \cup \mathcal{E}$, which is the result of applying the stratified Datalog program $\textbf{I}$ to $\mathcal{E}$.

Further, we call $\mathcal{E}$ an \textbf{external database} of the program structure~$\textbf{P}$ if it satisfies the constrains of $\textbf{P}$ after applying the Datalog program $\textbf{I}$ to $\mathcal{E}$, i.e. if
$$
\mathcal{E}^{\textbf{I}} \models \underset{\substack{\perp \leftarrow L_1,...,L_n \in \textbf{C} (\textbf{P}) \\
 \kappa \text{ interpretation on } \var (L_1,...,L_n)}}{\bigwedge}
\neg \left( \bigwedge_{i=1}^n L_i^{\kappa} \right)  
.$$ 
A \textbf{ground variable} is a ground atom~\mbox{$G := r (x_1,...,x_n)$} of $\mathfrak{Q}^*$ with a random predicate $r \in \mathfrak{Q}$. Finally, we write $\mathcal{G} (\mathcal{E})$ for the set of all ground variables. 
\label{notation - Herbrand model}
\end{definition}

The term ground variable indicates that we expect $G$ to denote a proper random variable under our semantics. From now on we restrict ourselves to the study of stratified program structures. Hence, let us fix a stratified program structure $\textbf{P}$ for the rest of this work.

\begin{definition}[FCM-Semantics of Lifted Programs]
Let $\pi$ be a choice of parameters and let~$\mathcal{E}$ be an external database for $\textbf{P}$. The \textbf{grounding}~$\textbf{P}^{(\pi, \mathcal{E})}$ of the program structure $\textbf{P}$ with respect to~$\pi$ and $\mathcal{E}$ is the system of Boolean equations given by   
\begin{equation}
G := \bigvee_{\substack{RC \in \textbf{R} (\textbf{P}) \\ \kappa \text{ interpretation on } \var (RC) \\ \effect (RC)^{\kappa} = G \\ \left( \mathcal{E}^{\textbf{I}} , \kappa \right) \models \condition(RC)}}
\left( 
\bigwedge_{C \in \causes (RC)} C^{\kappa} \wedge u \left( RC , \kappa \right) 
\right)
\label{The semantics of a lifted program}
\end{equation}
for every ground variable $G \in \mathcal{G} (\mathcal{E})$. Here, the \textbf{error term} $u \left( RC, \kappa \right)$ is a distinct Boolean random variable with the distribution 
$
\pi \left(u \left( RC, \kappa \right) \right) = \pi (RC) 
$
for every generalized ProbLog clause $RC \in \textbf{R} (\textbf{P})$ and every variable interpretation $\kappa$ on $\var(RC)$. Besides, the error terms $u \left( RC, \kappa \right)$ are assumed to be mutually independent. Finally, the \textbf{grounding} of the program $\textbf{Q} := \textbf{P}^{\pi}$ is given by~$\textbf{Q}^{\mathcal{E}} :=~\textbf{P}^{(\pi, \mathcal{E})}$.  
\label{Semantics of Programs}
\end{definition}

\begin{example}
It is easy to observe that the program structure $\textbf{P}$ of Example \ref{example - program structure} is indeed a stratified program structure in our sense. Now assume we are given a specific storage, which consists of four rooms \mintinline{prolog}{r1,r2,r3} and \mintinline{prolog}{r4}. These rooms are connected by passages as follows: 
\begin{minted}{prolog}
passage(r1,r2), passage(r2,r3)
\end{minted}
Moreover, we have five tanks~\mintinline{prolog}{t1,t2,t3,t4} and \mintinline{prolog}{t5} with 
\begin{minted}{prolog}
in(t1,r1), in(t2,r2), in(t3,r3), in(t4,r4), in(t5,r4).
\end{minted}
The tanks contain two types of liquids \mintinline{prolog}{gasoline} and \mintinline{prolog}{water}, which we describe in the following way: 
\begin{minted}{prolog}
stores(gasoline,t1), stores(gasoline, t2), stores(water, t3), 
stores(water, t4), stores(gasoline,t5), flammable(gasoline)
\end{minted} 
Finally, assume there are two employees Mary and John, which we express simply as~\mintinline{prolog}{mary} and~\mintinline{prolog}{john}. In this case, one checks that the storage above satisfies the integrity constraint~(\ref{example - integrity constaint}), i.e.~we are given an external database $\mathcal{E}$ for the program structure $\textbf{P}$. 

Further, let $\pi$ be the choice of parameters in Example \ref{example - program structure}. For $e \in \{ john , mary \}$, $t \in \{ t_1, t_2 , t_3 , t_4, t_5 \}$ and $r \in \{ r_1, r_2 , r_3, r_4 \}$ the grounding $\textbf{P}^{(\pi , \mathcal{E})}$ is given by the equations of the form:
\begin{enumerate}
\item[•] 
$opens (e,t) := u \left( (\_ :: opens (E,T) \leftarrow employee(E), tank(T).) , \{ E \mapsto e , T \mapsto t \} \right)$ such that \\ $u \left( (\_ :: opens (E,T) \leftarrow employee(E), tank(T).) , \{ E \mapsto e , T \mapsto t \} \right)$ is true with probability~$0.8$
\item[•]
$leaks (t) := opens(e,t) \land u \left( RC2 , \{ E \mapsto e, T \mapsto t \} \right)$ such that $u \left( RC2 , \{ E \mapsto e, T \mapsto t \} \right)$ is true with probability $0.1$
\item[•]
$smokes (e, r) := u \left( RC3 , \{ E \mapsto e, R \mapsto r \} \right)$ such that $u \left( RC3 , \{ E \mapsto e, R \mapsto r \} \right)$ is true with probability $0.5$
\item[•]
\mbox{$fire (r_i) := \bigvee_{j=1}^3 smokes (e, r_j) \land leaks(t_j) \land u (RC4, \{ R_1 \mapsto r_i, E \mapsto e, R \mapsto r_j, T \mapsto t_j \} )$} such that \\ \mbox{$u (RC4, \{ R_1 \mapsto r_i, E \mapsto e, R \mapsto r_j, T \mapsto t_j \} )$} is true with probability $0.05$ for \mbox{$1 \leq i,j \leq 3$}
\item[•]
$fire (r_4) := smokes (e, r_4) \land leaks(t_5) \land u (RC4, \{ R_1 \mapsto r_4, E \mapsto e, R \mapsto r_4, T \mapsto t_5 \} )$ such that \\ $u (RC4, \{ R_1 \mapsto r_4, E \mapsto e, R \mapsto r_4, T \mapsto t_5 \} )$ is true with probability $0.05$
\end{enumerate}
\label{example - grounding of a program structure}
\end{example}

In the present paper, we reason on a syntactic level about the conditional independencies implied by the program structure $\textbf{P}$ and an external database~$\mathcal{E}$. We proceed as Geiger and Pearl~\cite{OnTheLogicOfCausalModels} and restrict ourselves to those independence statements following from d-separation in the corresponding propositional causal structures, which we call ground graphs.

\begin{definition}[Ground Graph and Acyclicity]
Let~$\mathcal{E}$ be an external database for $\textbf{P}$. We define the \textbf{ground graph} $\Graph_{\mathcal{E}}(\textbf{P})$ to be the directed graph on the set of ground variables $\mathcal{G}(\mathcal{E})$, which is obtained by drawing an edge~$G_1 \rightarrow G_2$ if and only if there exists a generalized ProbLog clause $RC \in \textbf{R}(\textbf{P})$, a cause~$C \in \causes (RC)$ and a variable interpretation $\iota$ on $\var (C) \cup \var (\effect (RC))$ such that the following assertions are satisfied:
\begin{enumerate}
\item[i)]
$
\left( \mathcal{E}^{\textbf{I}}, \iota \right) \models \exists_{\var(\condition (RC)) \setminus \left( \var (C) \cup \var (\effect (RC)) \right)} \condition (RC)$
\item[ii)]  
$C^{\iota} \in \{ G_1 , \neg G_1 \}$ 
\item[iii)]
$\effect (RC)^{\iota} = G_2$.
\end{enumerate}
In this case we say that the edge $G_1 \rightarrow G_2$ is \textbf{induced} by $RC$.
Moreover, we call $\textbf{P}$ an \textbf{acyclic} program structure if $\Graph_{\mathcal{E}}(\textbf{P})$ is a directed acyclic graph i.e.~a causal structure for every external database $\mathcal{E}$.
\label{definition - Ground Graph}
\end{definition}  

From now on we assume the program structure $\textbf{P}$ to be acyclic. In this case it is easy to see that the grounding $\text{P}^{(\pi , \mathcal{E})}$ yields a unique expression for every ground variable~\mbox{$G \in \mathcal{G} (\mathcal{E})$} in terms of the error terms $u (RC, \kappa)$ for every choice of parameters $\pi$ and for every external database $\mathcal{E}$. Hence, it induces a unique probability distribution on~$\mathcal{G} (\mathcal{E})$. Rückschloß and Weitkämper \cite{fcm-semantics} further show that this distribution coincides with the distribution semantics of the ProbLog program~\mbox{$\textbf{P}^{\pi} \cup \mathcal{E}$}. Overall Definition~\ref{definition - program structure} yields a causal generalization of the standard semantics for acyclic ProbLog programs. Further, we obtain the following result:   

\begin{proposition}[Rückschloß and Weitkämper \cite{fcm-semantics}]
For every external database $\mathcal{E}$ and for every choice of parameters $\pi$ the grounding~$\textbf{P}^{(\pi, \mathcal{E})}$ yields a distribution, which is Markov to the ground graph $\Graph_{\mathcal{E}} (\textbf{P})$.~$\square$
\label{proposition - ground graph}
\end{proposition}

\begin{example}
The program structure $\textbf{P}$ of Example \ref{example - program structure} is indeed acyclic in our sense. Further, in the situation of Example \ref{example - grounding of a program structure} we obtain the following ground graph.
\newline
\adjustbox{scale=0.5,center}{
\begin{tikzcd}
                             & {smokes(mary,r_1)} \arrow[rdd] \arrow[rrdd] \arrow[rrrdd] & {smokes(john,r_1)} \arrow[rdd] \arrow[dd] \arrow[rrdd] & {smokes(mary,r_2)} \arrow[ldd] \arrow[dd] \arrow[rdd] & {smokes(john,r_2)} \arrow[dd] \arrow[ldd] \arrow[lldd] &                              &                               &                              \\
                             &                                                           &                                                        &                                                       &                                                        & {smokes(mary,r_3)}           & {smokes(mary,r_4)} \arrow[rd] & {smokes(john,r_4)} \arrow[d] \\
                             &                                                           & fire(r_1)                                              & fire(r_2)                                             & fire(r_3)                                              & {smokes(john,r_3)}           & {opens(john,t_5)} \arrow[rd]  & fire(r_4)                    \\
                             & leaks(t_1) \arrow[ru] \arrow[rru] \arrow[rrru]            &                                                        & leaks(t_2) \arrow[lu] \arrow[u] \arrow[ru]            & {opens(john,t_3)} \arrow[r]                            & leaks(t_3)                   & leaks(t_4)                    & leaks(t_5) \arrow[u]         \\
{opens(mary,t_1)} \arrow[ru] & {opens(john,t_1)} \arrow[u]                               & {opens(mary,t_2)} \arrow[ru]                           & {opens(john,t_2)} \arrow[u]                           & {opens(mary,t_3)} \arrow[ru]                           & {opens(mary,t_4)} \arrow[ru] & {opens(john,t_4)} \arrow[u]   & {opens(mary,t_5)} \arrow[u] 
\end{tikzcd}}
\label{example - grounded program}
\end{example}

As we now defined the necessary refinement of the ProbLog language \cite{Problog}, we can return to reasoning about conditional independence.

\section{A Symbolic Calculus for Conditional Independencies}
\label{sec - A Symbolic Calculus for Conditional Independencies}

By Proposition \ref{proposition - ground graph}, for every external database $\mathcal{E}$ the causal information about the conditional independencies implied by our program structure $\textbf{P}$ lies in the ground graph~$\Graph_{\mathcal{E}} (\textbf{P})$. Hence, applying Theorem~\ref{theorem - correctness of d-separation} we obtain a correct conditional independence oracle if we combine a Prolog representation of the ground graph $\Graph_{\mathcal{E}} (\textbf{P})$ with the d-separation oracle in Program \ref{program - dconnectes}.

Assume we are given a meta-predicate \mintinline{prolog}{random/2} that indicates all random predicates together with their arities. 

\begin{example}
In Example \ref{example - program structure}, this means we add the facts \\
\mintinline{prolog}{random(opens,2). random(leaks,1). random(smokes,2). random(fire,1).} \\ to the program structure $\textbf{P}$.
\end{example}

In this case, Program \ref{program - ground graph} computes the ground graph $\Graph_{\mathcal{E}} (\textbf{P})$ when applied to the program structure~$\textbf{P}$ and an external database $\mathcal{E}$.

\begin{program}[Representation of the Ground Graph]
$\text{ }$
\begin{minted}[mathescape]{prolog}
underlyingAtom(Atom,Atom) :- Atom \= (\+_).
underlyingAtom(Literal,Atom) :- Literal = (\+Literal1), 
	underlyingAtom(Literal1,Atom). 
% Determine the conditions in a body of a random clause  
conditions(Body,true) :- Body \= (_,_), underlyingAtom(Body,Atom), 
    functor(Atom,R,Arity), random(R,Arity).
conditions(Body,Body) :- Body \= (_,_), underlyingAtom(Body,Atom), 
    functor(Atom,R,Arity), \+random(R,Arity).
conditions((C1,Body), Cond) :- underlyingAtom(C1,Atom), 
    functor(Atom,R,Arity), random(R,Arity), conditions(Body,Cond).
conditions((C1,Body),(C1,Cond)) :- underlyingAtom(C1,Atom), 
    functor(Atom,R,Arity), \+random(R,Arity), conditions(Body,Cond).
% Calculate the potential parents
potentialParents(Body,[Atom]) :- Body \= (_,_), underlyingAtom(Body,Atom), 
    functor(Atom,R,Arity), random(R,Arity).
potentialParents(Body,[]) :- Body \= (_,_), underlyingAtom(Body,Atom), 
    functor(Atom,R,Arity), \+random(R,Arity).
potentialParents((C1,Body), [Atom|Parents]) :- underlyingAtom(C1,Atom), 
    functor(Atom,R,Arity), random(R,Arity), potentialParents(Body,Parents).
potentialParents((C1,Body), Parents) :- underlyingAtom(C1,Atom), 
    functor(Atom,R,Arity), \+random(R,Arity), 
    potentialParents(Body,Parents).
%Check whether edge exists
(X ---> Y) :- random(R,Arity), functor(Y,R,Arity), clause((_::Y),Body), 
	potentialParents(Body,Parents), member(X,Parents), 
	conditions(Body,Conds), Conds.
\end{minted}
\label{program - ground graph}
\end{program}

Together, Program \ref{program - dconnectes} and \ref{program - ground graph} yield a meta-interpreter that computes valid conditional independence statements implied by an acyclic, stratified program structure and an external database. In the appendix we prove the following result that establishes the completeness of this conditional independence oracle for a fragment of program structures.

\begin{theorem}[Completeness]
Let $\textbf{P}$ be a positive program structure and let $\mathcal{E}$ be an external database such that the ground graph $\Graph_{\mathcal{E}}(\textbf{P})$ is singly connected. Further, let $\pi$ be a  choice of parameters for~$\textbf{P}$ with values in $(0,1)$ that yields proper unconditional probabilities for all ground variables. In this case, the grounding $\textbf{P}^{(\pi, \mathcal{E})}$ yields a distribution that is faithful to the ground graph~\mbox{$\Graph_{\mathcal{E}}(\textbf{P})$}.
 
In particular, if the ground graph $\Graph_{\mathcal{E}}(\textbf{P})$ is singly connected for every external database $\mathcal{E}$ and has a generalized ProbLog clause grounding to a probabilistic fact for every source in~$\Graph_{\mathcal{E}}(\textbf{P})$, our meta-interpreter is correct and complete for every choice of parameters~$\pi$ of $\textbf{P}$ with values in~$(0,1)$.  
\label{theorem - completeness}
\end{theorem}

As singly connected ground graphs imply mutual independence of body atoms, the fragment of 
Theorem \ref{theorem - completeness} reminds us of the (Ind,Ind) assumption 
\cite{PITAOPT}, under which marginal probabilities can be computed in a much simpler way.

\begin{example}
As the program structure $\textbf{P}$ of the introduction lies in the fragment of Theorem \ref{theorem - completeness}, Program~\ref{program - dconnectes} and \ref{program - ground graph} yield a correct and complete conditional independence oracle whenever we choose probabilities in $(0,1)$ for every generalized ProbLog clause in $\textbf{P}$.   

%
\end{example} 

\section{Experimental Evaluation} 

To evaluate our conditional independence oracle in Program \ref{program - dconnectes} and \ref{program - ground graph}, we investigate the lifted program $\textbf{P}$  that defines the random predicate $random(p,1)$ with the random part below. 
\begin{align*}
& random(p,1). &
0.5 :: p(X) \leftarrow n(X). &&
0.5 :: p(Y) \leftarrow  p(X), n(X), n(Y), e(X,Y).
\end{align*}
Further, the program $\textbf{P}$ expects a directed acyclic graph $G$ as external database, represented by storing its nodes~$n/1$ and its edges $e/2$. Here, $p/1$ is a random property that holds with a base probability of~$0.5$ for every node of $G$. If we observe $p(n)$ for a node $n$ of $G$, this enhances the probability of $p(c)$ for every child $c$ of $n$ by a factor of $0.5$. 

Further, we generate directed acyclic graphs $G$ with $S := 5i$ nodes for $1 \leq i \leq 20$ by sampling five times from the following ProbLog program. 
\begin{align*}
&n(1). &
n(Y) \leftarrow n(X),~Y=X+1,~X<S+1.
&&  \dfrac{1}{\sqrt{S}}  :: e(X,Y) \leftarrow n(X),~n(Y),~X<Y.
\end{align*}
In this way we obtain five directed acyclic graphs on the nodes \mbox{$\{ 1,...,S \}$} that 
have an edge $i \rightarrow j$ with probability $\sqrt{S}^{-1}$ for every~\mbox{$1 \leq i < j \leq S$}. Next, we choose for every graph size $S$ ten tuples~$(a,b)$ of even numbers between 
one and $S$. We then process the queries 
$$
dseparates(p(a),p(b),[]) \text{ and } dseparates(p(a),p(b),[p(i):~i\text{ odd number between 1 and S}])
$$ on~$\textbf{P}$ for every graph of size $S$ using the meta-interpreter of Program \ref{program - dconnectes} and \ref{program - ground graph} with a timeout of $10$ seconds. 

Finally, we aim to compare our approach with checking the definition of independence, i.e.~with checking whether $\pi(p(a) \land p(b)) = \pi(p(a)) \cdot \pi(p(b))$. To this aim we additionally calculate the probabilities of $p(a)$, $p(b)$ and $p(a) \land p(b)$ with the evidence $\{ p(i):~i \text{ odd number between 1 and S} \}$ and without evidence using ProbLog 2 with a timeout of $10$ seconds. The median and maximal run times of the queries described above are visualized in Figure \ref{fig:scale}, clearly demonstrating that our approach performs significantly faster than checking (conditional) independencies with exact inference in ProbLog~2.

\begin{figure}[tbph]
    \centering
    \includegraphics[width=0.49\textwidth]{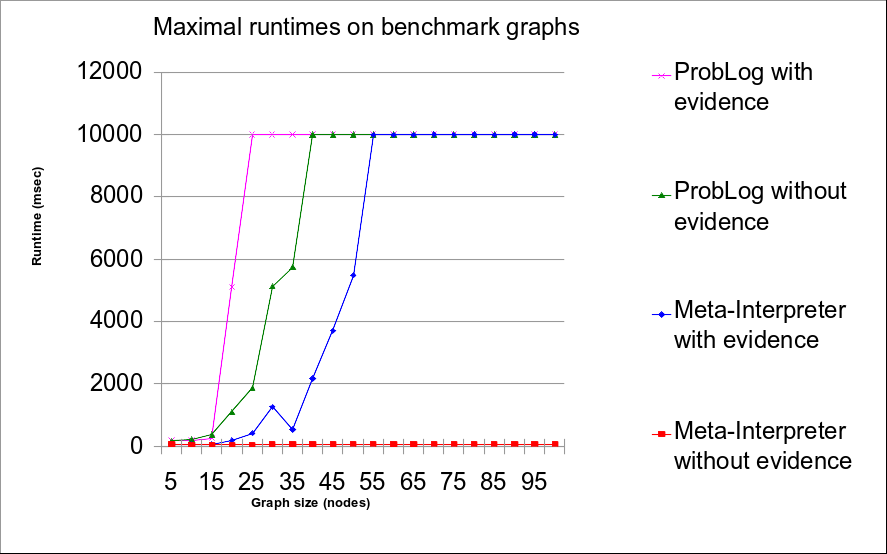}
    \includegraphics[width=0.49\textwidth]{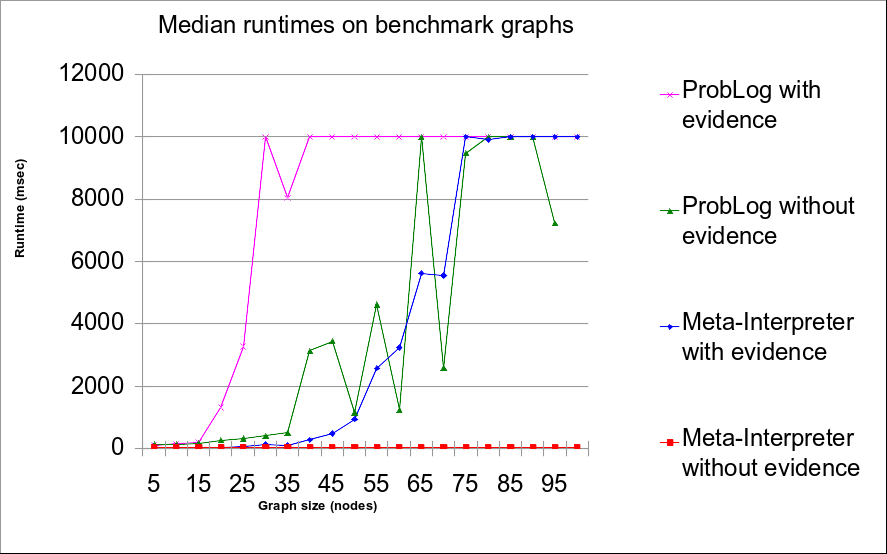}
    \caption{}
    \label{fig:scale}
\end{figure} 
 
\section{Conclusion}
First, we establish a framework for lifted probabilistic logic programming, i.e.~we abstract the notion of a ProbLog program away from a concrete database and away from the concrete probabilities mentioned in the clauses. In this way we obtain so-called program structures. As our main result we then use the theory of d-separation from Bayesian networks to reason about conditional independence on the basis of these program structures. We also implement the corresponding independence oracle as a meta-interpreter in Prolog. Finally, we prove the completeness of our conditional independence oracle for a fragment of programs structures in Theorem \ref{theorem - completeness}. The paper then closes with an experimental evaluation of our results revealing that our approach processes significantly faster than checking the definition of independence with exact inference in ProbLog 2.

As the theory of d-separation is the basis of causal structure discovery in Bayesian networks, one direction for future work is to develop the analogue of this theory in probabilistic logic programming. In this context, we note that causal faithfulness is needed to extract possible causal structures from an observed distribution. Furthermore, Theorem~\ref{theorem - obstruction to causal faithfulness} suggests that our independence oracle is complete for most program structures and choices of parameters.  Hence, in our opinion determining completeness results for more general fragments of program structures is a promising direction for future work.

\bibliographystyle{eptcs}

\bibliography{literature}

\section*{Appendix}

\begin{lemma}[Weak Transitivity, $\text{\cite[p.137, Exercise 3.10]{BNPearl}}$]
Let $\textbf{A}$ and $\textbf{C}$ be sets of random variables. Further, let $B$ be a Boolean random variable and let $\pi$ be a distribution on $\textbf{A} \cup \textbf{C} \cup \{ B \}$.
Choose possible values $\textbf{a}$ and $\textbf{c}$ for the random variables in $\textbf{A}$ and $\textbf{C}$ respectively. Further, choose a possible value $b$ for~$B$. If $b$ is both dependent on $\textbf{a}$ and on $\textbf{c}$, we find that $\textbf{a}$ and $\textbf{c}$ are marginally dependent or they become dependent once we condition on $b$.~$\square$   
\end{lemma}

\begin{lemma}\label{propdep}
  Let $\pi$ be a distribution on the set $\Omega$ of all valuations on the variables $X,X_1,\dots,X_n$ such that $\{X\} \cup \{X_i\}_{i=1,\dots,n}$ encodes a mutually independent set of events  and let $\varphi$ be a propositional formula (in minimal DNF) in which $X$ occurs, and only occurs positively.

  Then $\varphi$ and $X$ are positively correlated as events on $\Omega$.
\end{lemma}

\begin{proof}
  We must show that $\pi (\varphi\mid X) > \pi(\varphi)$.
  Let $\psi$ be the formula obtained by substituting $\top$ for $X$ in~$\varphi$.
  Then we have the following: \[\pi(\varphi\mid X) = \pi(\psi\mid X) = \pi(\psi).\]
  The first equality holds because $\varphi$ and $\psi$ are equivalent on valuations satisfying $X$, and the second equality holds because $\pi$ is an independent distribution.
  As $\psi$ is staisfied by all valuations satisfying~$\varphi$, but also on at least one valuation not satisfying $\varphi$ (since $X$ occurs in the minimal DNF formula $\varphi$),
  this implies the statement. 
\end{proof}

\begin{proof}[Proof of Theorem \ref{theorem - completeness}]
Let $G := \Graph_{\mathcal{E}}(\textbf{P})$ be the ground graph of the program structure $\textbf{P}$ with respect to the external database $\mathcal{E}$.
Let $P$ be a d-connecting path from $A$ to $B$ over a subset of nodes $\textbf{Z}$ in $G$.
We show that for all truly probabilistic choices of parameters, $A$ and $B$ are dependent over $\textbf{Z}$.
Let $\pi$ be the probability distribution on the random variables corresponding to such a choice of parameters. 
We proceed by induction on the length of~$P$.

As we use this in the argument for colliders below,
we first show that for directed paths, one obtains the stronger statement that $A$ and $B$ are positively correlated over $\textbf{Z}$.

We begin with paths of length two.
If the length of $P$ is two, there is an edge between~$A$ and $B$. Assume without loss of generality that the edge goes from $A$ to $B$. 
Consider the definition $\varphi$ of $B$ in the Clark completion of the logic program associated with $\textbf{P}$.
This is a propositional formula in the parents of $B$ and the error terms of ground clauses with head $B$;
Since $A$ is a parent of $B$, there is a clause $\mathcal{C}$ whose body includes $A$,
and thus since $B$ is true if and only if $A$ is true when the error terms associated with all other clauses are set to false and all other facts in the body of $\mathcal{C}$ are true,
$A$ occurs in a minimal disjunctive normal form representation of $\varphi$.
As $G$ is singly connected, the only path between any two parents of $B$ is via $B \notin \textbf{Z}$,
and thus the distribution induced by~$\pi\mid Z$ on the parents of $B$ is independent by the correctness of d-separation.
Therefore~$A$ and $B$ are positively correlated by Lemma \ref{propdep}.

So assume it to be true for directed paths of length $n$ and let $A_1\dots A_n A_{n+1}$ be a directed path of length~$n+1$.
Then $\pi(A_{n+1}\mid A_1,Z)$ is equal to \[\pi(A_{n+1}\mid A_n,A_1,Z) \pi(A_n\mid A_1,Z) + \pi(A_{n+1}\mid \neg A_n,A_1,Z) \pi(\neg A_n\mid A_1,Z)\]
and  $\pi(A_{n+1}\mid Z)$ is equal to \[\pi(A_{n+1}\mid A_n,Z) \pi(A_n\mid Z) + \pi(A_{n+1}\mid \neg A_n,Z) \pi(\neg A_n\mid Z). \]
Since $P$ is the only path between $A_{1}$ and $A_{n+1}$, $A_{1}$ and $A_{n+1}$ are d-separated by $A_n$.
Thus, we find \mbox{$\pi(A_{n+1}\mid A_n,A_1,Z) = \pi(A_{n+1}\mid A_n,Z)$} and~\mbox{$\pi(A_{n+1}\mid  \neg A_n,A_1,Z) = \pi(A_{n+1}\mid \neg A_n,Z)$}.
Furthermore, by the induction hypothesis, $\pi(A_{n+1}\mid A_n,Z) > \pi(A_{n+1}\mid \neg A_n,Z)$ and~\mbox{$\pi(A_n\mid A_1,Z) > \pi(A_n\mid Z)$}.
Overall, this implies that $\pi(A_{n+1}\mid A_1,Z) > \pi(A_{n+1}\mid Z)$ as required. \hfill($\square$)

We now return to the case of a general d-connecting path.
As paths of length 2 are necessarily directed, the base step of the induction follows from the directed case above.
So assume that dependence holds for all paths of length $n$ and let $P$ be a path of length~\mbox{$n+1 > 2$}.

We proceed by appeal to weak transitivity.

Let $A'$, $C$ and $B$ be the final three nodes in $P$.
We first cover the case where $C$ is not a collider.

Then by the induction hypothesis, $A$ and $C$ are dependent over $\textbf{Z}$,
and by the length two case above, $C$ and $B$ are dependent over $\textbf{Z}$.
Thus, either $A$ and $B$ are dependent over $\textbf{Z}$,
or they are dependent over~$\{ C \} \cup \textbf{Z}$.
However, since $G$ is singly connected, $P$ is the only path from $A$ to $B$
and thus $A$ and $B$ are d-separated by~$\{ C \} \cup \textbf{Z}$.
Together with the correctness of d-separation we can conclude that $A$ and $B$ are dependent over $\textbf{Z}$.

Now consider the case of a collider where $C$ has a descendant in $\textbf{Z}$.
Note that since conditioning on a descendant $C_i$ of $C$ blocks the (only) path between $A$ and $B$ and any descendant $C'$ of $C_i$,
we can assume by correctness of d-separation without loss of generality that all descendants $C_i$ of $C$ in $\textbf{Z}$ are non-descendants of each other.
Let~\mbox{$\{C_i\}_{i=1,\dots,n}$} be the descendants of $C$ in $\textbf{Z}$, and let $\textbf{Z}' := \textbf{Z} \setminus \{C_i\}_{i=1,\dots,n}$.

We show by induction on $n$ that $\pi(X \mid C_1,\dots,C_n,\textbf{Z}') > \pi(X,\textbf{Z}')$ for~\mbox{$X\in \{A' ,B\}$}.
The base case of the induction is given by the special case of a directed path above.
So assume \mbox{$\pi(X \mid C_1,\dots,C_n,\textbf{Z}') > \pi(X\mid \textbf{Z}')$}. We want to show that~\mbox{$\pi(X \mid C_1,\dots,C_n,C_{n+1},\textbf{Z}') > \pi(X\mid \textbf{Z}')$}.
However, we know that
$$
\pi(X \mid C_1,\dots,C_n,C_{n+1},\textbf{Z}') > \pi(X \mid C_1,\dots,C_n,\textbf{Z}')
$$ 
by the special case of directed paths above, from which the claim follows. \hfill($\square$)

To apply weak transitivity, we alter the graph by introducing a new node $\bigwedge C$ with arrows from every~$C_i$ into  $\bigwedge C$.
We extend the probability distribution to $\bigwedge C$ by setting~\mbox{$\bigwedge C$} to be true of and only if all $C_i$ are true.
By the argument above, $A'$ and $B$ are both positively correlated with  $\bigwedge C$ over $\textbf{Z}'$.
Thus, by weak transitivity, either $A'$ and $B$ are dependent over $\textbf{Z}'$ or they are dependent over $\{\bigwedge C\} \cup \textbf{Z}'$.
The former is excluded by the correctness of d-separation and the fact that the collider $C$ blocks the only path from  $A'$ to $B$ (since the original graph was singly connected).
Therefore  $A'$ and $B$ are dependent over $\{\bigwedge C\} \cup \textbf{Z}'$, which implies  $A'$ and $B$ being dependent over $\textbf{Z} = C_1, \dots,C_n,\textbf{Z}'$ in the original graph.

We can now conclude the proof by using weak transitivity and singly-connectedness a final time to deduce that $A$ and $B$ are dependent over $Z$ from
the fact that $A$ and $A'$ are dependent over $Z$ (the induction hypothesis)
and the fact that $A'$ and $B$ are dependent over $Z$.
\end{proof}

\end{document}